\documentclass[a4paper,11pt]{article}

\usepackage[margin=28mm]{geometry}
\usepackage{amsmath,amssymb,amsfonts,amsthm}
\usepackage{graphicx}
\usepackage{booktabs,url,placeins}
\numberwithin{equation}{section}

\newcommand{\R}{\mathbb{R}}
\newcommand{\C}{\mathbb{C}}

\newcommand{\diag}{\operatorname{diag}}
\newcommand{\tr}{\operatorname{tr}}
\newcommand{\Arg}{\operatorname{Arg}}

\renewcommand{\Im}{\operatorname{Im}}
\renewcommand{\Re}{\operatorname{Re}}

\newtheorem{thm}{Theorem}[section]
\newtheorem{lem}[thm]{Lemma}
\newtheorem{cor}[thm]{Corollary}
\newtheorem{prop}[thm]{Proposition}

\theoremstyle{definition}

\theoremstyle{remark}
\newtheorem{rem}[thm]{Remark}

\title{Geometry of Critical Couplings, Quartic Binding, and Spectral Shift for Concentric $\delta$-Shells}
\author{Masahiro Kaminaga}
\date{}

\begin{document}

\maketitle

\noindent\textbf{Keywords:} $\delta$-shell interactions; critical coupling surface; threshold resonance; quartic binding; spectral shift function; inverse problem

\noindent\textbf{Mathematical Subject Classification (2020):} 35J10, 35P25, 47A40, 81U20

\begin{abstract}
We study the first zero-energy resonance of finitely many attractive concentric spherical $\delta$-shells with fixed radii and variable strengths.
The positive critical strengths form a smooth hypersurface in coupling space.
For every nonzero tangent redistribution we calculate its second fundamental form as an explicit strictly positive quadratic form $C_a$ in cumulative weighted coupling changes.
Along the straight perturbation $g_*+\varepsilon a$, the same coefficient determines the unique emerging bound state, $E_b=-C_a^2\varepsilon^4+O(|\varepsilon|^5)$, for either sign of $\varepsilon$, and the crossover of the canonically normalized total spectral shift on the scale $k\asymp\varepsilon^2$.
We also give a global smooth parametrization of the positive critical surface by probability weights, determine the exact attainable interval of effective ranges, and reconstruct all critical strengths from individual threshold-pole responses.
An exact coupling coordinate gives a uniform $O(k^3)$ spectral-shift normal form and a differentiated $O(k^2)$ estimate without a restriction on the ratio of coupling distance to wave number.
The scattering determinant and general threshold exponents are established inputs.
Our results give the model-specific curvature coefficient, its spectral consequences, and a complete local and global description of the first-critical configurations.
\end{abstract}

\section{Introduction}

Consider the Schr\"odinger operator
\begin{equation}\label{eq:intro-H}
 H_g=-\Delta-\sum_{j=1}^Ng_j\delta(|x|-R_j),
 \qquad 0<R_1<\cdots<R_N,\quad g_j>0,
\end{equation}
with fixed shell radii and variable attractive strengths.
The first bound state appears at a simple $s$-wave zero-energy resonance.
For one strength parameter, this passage is described by the usual threshold law.
For several strengths, the set of critical couplings is a hypersurface, and the behavior of perturbations tangent to this hypersurface becomes an additional question.
We study its geometry together with the low-energy spectral response.

The self-adjoint realization, boundary resolvent, and zero-energy matrix criteria for concentric shells were obtained by Kaminaga \cite{KaminagaAAMP2026}.
Earlier versions of the present paper established the channel scattering determinant identity and studied the $s$-wave scattering length and threshold behavior for two shells.
These results are retained as preliminary results and are not counted as new contributions of the present version.
Earlier work on concentric shells includes Shabani \cite{Shabani1988} and Hounkonnou, Hounkpe, and Shabani \cite{HounkonnouHounkpeShabani1997}.
For general background on solvable models with singular interactions, we refer to Albeverio et al. \cite{AlbeverioEtAl2005}.
Scattering theory for penetrable wall potentials was studied by Ikebe and Shimada \cite{IkebeShimada1991}, and Shimada investigated the corresponding low-energy scattering problem \cite{Shimada1994}.
General scattering theory for singular hypersurface interactions was developed by Behrndt, Malamud, and Neidhardt \cite{BehrndtMalamudNeidhardt2017} and by Mantile, Posilicano, and Sini \cite{MantilePosilicanoSini2018,MantilePosilicano2019}.
Behrndt, Gesztesy, and Nakamura gave formulas for spectral shift functions in terms of boundary operators \cite{BehrndtGesztesyNakamura2018}.
Klaus, Simon, and Golovaty studied coupling thresholds, including higher-order binding when a first-order perturbation term vanishes \cite{KlausSimon1980,Klaus1982,Golovaty2020}.
Thus the determinant identity, the two-shell threshold condition, and the fourth-order binding exponent alone are not new claims of this paper.

Our main object is the first-critical hypersurface
\begin{equation}\label{eq:intro-critical-set}
 \mathcal C=\{g\in(0,\infty)^N:\rho(g)=1\},
 \quad
 \rho(g)=\lambda_{\max}(B(g)),\qquad
 B(g)_{ij}=\sqrt{g_i g_j}\min(R_i,R_j).
\end{equation}
At $g_*\in\mathcal C$, let $\phi$ be the positive reduced zero-energy resonance normalized by $\phi(R_N)=1$ and put $\phi_j=\phi(R_j)$.
The tangent space is characterized by $\sum_j a_j\phi_j^2=0$.
For each nonzero tangent vector, we obtain the quadratic form
\begin{equation}\label{eq:intro-curvature}
 C_a=\sum_{j=1}^{N-1}\frac{R_{j+1}-R_j}{\phi_j\phi_{j+1}}
              \left(\sum_{i=1}^{j}a_i\phi_i^2\right)^2>0.
\end{equation}
Theorem~\ref{thm:tangent-binding} identifies this form with the second fundamental form of $\mathcal C$, with the outward-normal convention stated there.
In particular, the critical surface has strictly positive curvature in every nonzero tangent direction.
The same theorem gives, for the straight path $g_*+\varepsilon a$,
\begin{equation}\label{eq:intro-binding}
 E_b(\varepsilon)=-C_a^2\varepsilon^4+O(|\varepsilon|^5),
 \qquad
 \Delta N(k^2;g_*+\varepsilon a)
 =\frac12+\frac1\pi\arctan\frac{C_a\varepsilon^2}{k}
   +O(|\varepsilon|+k).
\end{equation}
Here $\Delta N=-\xi$ is the normalized spectral excess and $\xi$ is the spectral shift function relative to the free operator.
The strict positivity holds in every nonzero tangent direction, with no exceptional redistribution.
Convexity of the largest eigenvalue of a linearly parameterized symmetric matrix explains why the Hessian cannot be negative, but does not provide the explicit coefficient or the spectral consequences.
The path-graph representation in the proof identifies the coefficient as an energy of cumulative coupling changes.
This coefficient determines the outward second fundamental form, the amount by which a straight tangent line crosses the critical surface, the leading binding energy, and the spectral-shift crossover.
It therefore gives a quantitative relation between geometry and threshold spectra, rather than a new universal threshold exponent.

There is also a global description of the same critical hypersurface.
For positive probability weights $w_j$ with $\sum_jw_j=1$, define
\begin{equation}\label{eq:intro-weight-map}
 \phi_w(r)=\sum_{j=1}^Nw_j\min\{r/R_j,1\},
 \qquad
 g_j(w)=\frac{w_j}{R_j\phi_w(R_j)}.
\end{equation}
Theorem~\ref{thm:inverse-critical} proves that this map parametrizes every positive first-critical configuration exactly once.
It determines the full attainable range of effective ranges, not only bounds for individual configurations.
The effective-range function is strictly concave in the weights, and the critical strengths can be recovered from individual pole responses.
Corollary~\ref{cor:geometry-compatibility} connects the acceleration of curves in these global coordinates to the local curvature form $C_a$.
Thus the local and global results describe the same critical coupling surface from complementary viewpoints.
The positive probability weights are global coordinates, whereas the quadratic form measures the correction needed to remain critical when those coordinates are varied.
The global inverse map and its coefficient-level relation to the local curvature form are established in the present version and were not part of the earlier scattering analysis.

To connect this geometry to scattering observables, we use an exact coupling coordinate.
Let $R=R_N$ and let $y_0$ solve the interior zero-energy radial problem, including the first $N-1$ shells, with $y_0(R)=1$.
Set
\begin{equation}\label{eq:intro-DN-coordinate}
 p(g)=g_N-y_0'(R-),\qquad
 \mathcal I(g)=\int_0^R y_0(r)^2\,dr.
\end{equation}
The coordinate $p$ vanishes exactly on $\mathcal C$ near $g_*$ and is positive on the bound-state side.
Theorem~\ref{thm:exact-normal-form} gives
\begin{equation}\label{eq:intro-cubic}
 \Delta N(k^2;g)=\frac12+\frac1\pi\arctan\left(\frac{p(g)}k
                  +\mathcal I(g)k\right)-\frac{Rk}{\pi}+O(k^3),
\end{equation}
uniformly on a full neighborhood of $g_*$, without any condition on $p(g)/k$.
The derivative of the remainder with respect to $k$ is uniformly $O(k^2)$.
These estimates concern the normalized sum over all partial waves, not only a fixed $s$-wave phase shift.
The complementary Birman-Schwinger coordinate and its explicit Schur-complement coefficients are given in Theorem~\ref{thm:general-threshold}.

The scope of the new results should be distinguished from standard consequences of an attractive first resonance.
Bethe's effective-range identity \cite{Bethe1949} and the general finite-range bound of Hammer and Lee \cite{HammerLee2010} are established.
Positivity and concavity of the first radial resonance also lead to effective-range inequalities beyond the shell model.
What is specific here is the exact realization of all intermediate effective ranges with strictly positive shell strengths, together with the full inverse map and the strictly positive tangent curvature formula.
Deformations of the supporting surface, as considered by Exner and Fraas \cite{ExnerFraas2009}, are different from the fixed-radius strength redistributions studied here.

Section~\ref{sec:model} recalls the previously established boundary matrices.
Sections~\ref{sec:scatt} and \ref{sec:shift} record the determinant formula and establish the normalization needed for the total spectral shift.
Section~\ref{sec:general-first-threshold} proves uniform low-energy estimates.
The principal new local and global geometric results are in Section~\ref{sec:critical-geometry}.
Section~\ref{sec:numerical-main} checks their coefficients and the uniform phase normal form.
Appendix~\ref{app:normalization} gives details of the canonical spectral-shift identification.

\section{The model and its finite boundary matrices}\label{sec:model}

Let $0<R_1<\cdots<R_N$, $S_j=\{x:|x|=R_j\}$, and $\alpha_j\in\R$.
We use the self-adjoint realization of $H=-\Delta+\sum_j\alpha_j\delta(|x|-R_j)$ defined by the form
\begin{equation}\label{eq:form}
h[u,v]=\int_{\R^3}\nabla u\cdot\nabla\overline v\,dx
 +\sum_{j=1}^N\alpha_j\int_{S_j}u\overline v\,d\sigma_j,
\qquad D[h]=H^1(\R^3).
\end{equation}
The trace inequality makes the boundary terms infinitesimally form bounded.
The form is therefore closed and lower semibounded.
This construction follows the quadratic-form treatment of singular interactions by Brasche, Exner, Kuperin, and \v{S}eba \cite{BrascheEtAl1994} and the hypersurface formulation of Behrndt, Langer, and Lotoreichik \cite{BehrndtLangerLotoreichik2013}.
For the boundary resolvent method, we also refer to the abstract construction of Posilicano \cite{Posilicano2001}.

Put $H_0=-\Delta$.
For $z\in\C\setminus[0,\infty)$, choose $\Im\sqrt z>0$ and set
\[
G_z(x,y)=\frac{e^{i\sqrt z|x-y|}}{4\pi|x-y|},\qquad
(\Gamma_j(z)\varphi)(x)=\int_{S^2}G_z(x,R_j\omega)\varphi(\omega)\,d\omega.
\]
Our boundary space is $\mathcal G=\bigoplus_{j=1}^N L^2(S^2)$, with $\Gamma(z)\varphi=\sum_j\Gamma_j(z)\varphi_j$.
Define $m(z)$ on $\mathcal G$ by
\[
(m_{ij}(z)\varphi)(\omega)
 =\int_{S^2}G_z(R_i\omega,R_j\omega')\varphi(\omega')\,d\omega',
\]
and put
\begin{equation}\label{eq:Theta}
\Theta=\diag(\alpha_1R_1^2,\ldots,\alpha_NR_N^2),\qquad
K_N(z)=I+m(z)\Theta.
\end{equation}
The factors $R_j^2$ occur because the boundary integrals above use $d\omega$, not surface measure on $S_j$.

Kaminaga established the following realization and boundary resolvent formula for concentric shells \cite{KaminagaAAMP2026}.
We state the constant-strength case in the notation used here.

\begin{thm}[Realization and resolvent formula]\label{thm:input}
The operator associated with \eqref{eq:form} acts as $-\Delta$ away from the shells.
Its domain consists of the $L^2$ functions which are $H^2$ on each region, are continuous across the shells, and satisfy, in the trace sense,
\begin{equation}\label{eq:jump-intro-model}
\partial_r u(R_j+0,\omega)-\partial_r u(R_j-0,\omega)
 =\alpha_j u(R_j,\omega),\qquad j=1,\ldots,N.
\end{equation}
For $z\in\C\setminus[0,\infty)$, $K_N(z)$ is noninvertible precisely when $z$ is an eigenvalue of $H$.
Otherwise,
\begin{equation}\label{eq:resolvent}
(H-z)^{-1}=(H_0-z)^{-1}
 -\Gamma(z)\Theta K_N(z)^{-1}\Gamma(\overline z)^*.
\end{equation}
The resolvent difference is trace class.
\end{thm}

For completeness, the trace-class property can also be checked directly: each $\Gamma_j(z)$ is Hilbert-Schmidt, since
\[
\int_{\R^3}|G_z(x,y)|^2\,dx
 =\frac{1}{8\pi\Im\sqrt z},
\]
and the boundary has finite measure.
The correction in \eqref{eq:resolvent} is the product of two Hilbert-Schmidt operators and a bounded operator.

Let $Y_{\ell m}$ be an orthonormal basis of spherical harmonics, and write $\mathcal H_\ell=\operatorname{span}\{Y_{\ell m}:-\ell\le m\le\ell\}$.
Constant strengths preserve these spaces.
We denote the reduced $N\times N$ matrices by $m_\ell(z)$ and $K_\ell(z)$.

\begin{lem}\label{lem:mell}
For $k=\sqrt z$ with $\Im k>0$,
\begin{equation}\label{eq:mell-entry}
(m_\ell(z))_{ij}
 =ik\,j_\ell(kR_<)h_\ell^{(1)}(kR_>),
\end{equation}
where $R_<=\min(R_i,R_j)$ and $R_>=\max(R_i,R_j)$, and
\begin{equation}\label{eq:Kell}
K_\ell(z)=I_N+m_\ell(z)\Theta.
\end{equation}
Here $j_\ell$ and $h_\ell^{(1)}$ are the spherical Bessel and outgoing Hankel functions.
We write $h_\ell^{(2)}$ for the incoming Hankel function.
\end{lem}

\begin{proof}
The spherical harmonic expansion of the free Green function is
\[
G_z(x,y)=ik\sum_{n=0}^\infty\sum_{q=-n}^n
j_n(kr_<)h_n^{(1)}(kr_>)Y_{nq}(\widehat x)
\overline{Y_{nq}(\widehat y)}.
\]
Projection onto $Y_{\ell m}$ gives \eqref{eq:mell-entry}.
Since $\Theta$ acts only on the shell index, \eqref{eq:Kell} follows.
\end{proof}

\section{The channel scattering determinant}\label{sec:scatt}

The channel determinant formula was derived in an earlier version of the present paper, using the boundary resolvent of \cite{KaminagaAAMP2026}.
We recall its statement and a short verification because the continuous choice of phase is needed in the next section.
The general stationary scattering theory is provided by \cite{BehrndtMalamudNeidhardt2017,MantilePosilicanoSini2018,MantilePosilicano2019}.
We use the Fourier convention of \cite[Ch.~XI]{ReedSimonIII}.

\begin{prop}[Channel scattering determinant]
\label{thm:Sell}\label{prop:Kell-invertible-positive}
For every $\ell\ge0$ and $k>0$, $K_\ell^+(k):=K_\ell(k^2+i0)$ is invertible.
The scattering fiber has the smooth representative
\begin{equation}\label{eq:Sell-det}
 S_\ell(k)=\frac{\det K_\ell(k^2-i0)}{\det K_\ell(k^2+i0)},
 \qquad |S_\ell(k)|=1.
\end{equation}
\end{prop}
\begin{proof}
The outgoing single-layer solution associated with a vector $c$ in the kernel of $K_\ell^+(k)$ is regular at the origin, satisfies the transmission conditions, and is outgoing outside the last shell.
Green's identity and real coupling strengths show that its outgoing flux is zero.
The outgoing Hankel asymptotics therefore force the exterior coefficient to vanish.
Continuation of the radial Cauchy data across the shells gives $c=0$.
This proves invertibility.

Put $b_j=j_\ell(kR_j)$ and $c=(K_\ell^+)^{-1}b$.
The regular solution constructed with the single-layer potential has exterior form
\[
 \frac12\bigl\{h_\ell^{(2)}(kr)+\sigma_\ell(k)h_\ell^{(1)}(kr)\bigr\}Y_{\ell m},
 \quad
 \sigma_\ell(k)=1-2ikb^T\Theta(K_\ell^+)^{-1}b.
\]
Outgoing uniqueness identifies $\sigma_\ell$ with the scattering eigenvalue in the stationary theory of \cite[Theorem~5.6 and Sec.~6.4]{MantilePosilicanoSini2018}.
The Bessel identities give $m_\ell^- -m_\ell^+=-2ikbb^T$.
The matrix determinant lemma then gives $\det K_\ell^- =\det K_\ell^+\sigma_\ell$.
Since the coefficients are real, $\det K_\ell^-=\overline{\det K_\ell^+}$, and the result follows.
\end{proof}

The original scattering fibers are defined almost everywhere in energy.
We use the smooth representative in \eqref{eq:Sell-det} for every $k>0$.
Write $\mathsf S(k)=S(k^2)$ to distinguish wave number from energy.

\section{Spectral shift and the normalized Friedel sum}\label{sec:shift}

The scattering determinant is an established representation and is not counted as an independent new result here.
Its normalization is nevertheless needed to compare the critical configurations with the full spectral-shift response.
The determinant determines $\delta_\ell(k)$ modulo $\pi$.
For a normalized spectral shift, the integer part of the phase is also needed.
We fix it by coupling continuation.
Throughout this section, $K_\ell(k)$ denotes $K_\ell(k^2+i0)$ and $D_\ell(k)=\det K_\ell(k)$.

\subsection{Phase normalization by coupling continuation}
\label{sec:phase-normalization}

A principal argument of $D_\ell(k)$ does not in general give the normalized phase.
To specify the branch, put
\begin{equation}
\label{eq:coupling-K}
K_{\ell,t}(k)=I_N+tm_\ell(k^2+i0)\Theta,\qquad 0\le t\le1,
\end{equation}
and define
\begin{equation}
\label{eq:coupling-log}
\mathcal L_\ell(k)
 =\int_0^1\tr\{K_{\ell,t}(k)^{-1}\partial_tK_{\ell,t}(k)\}\,dt,
\qquad \delta_\ell(k)=-\Im\mathcal L_\ell(k).
\end{equation}
This is the logarithm of $D_\ell(k)$ continued from $\log1=0$.
It satisfies $e^{\mathcal L_\ell(k)}=D_\ell(k)$ and $S_\ell(k)=e^{2i\delta_\ell(k)}$.
For each fixed $\ell$, it also gives $\delta_\ell(k)\to0$ as $k\to\infty$, since $m_\ell(k^2+i0)=O(k^{-1})$ in that limit.
This last statement is channelwise; it is not a normalization of the infinite total phase at high energy.
Appendix~\ref{app:normalization} shows that \eqref{eq:coupling-log} agrees with the canonical spectral shift normalization.

There is a convenient finite-matrix implementation.
If $\mu_{\ell,j}(k)$, $1\le j\le N$, are the eigenvalues of $m_\ell(k^2+i0)\Theta$, counted with algebraic multiplicity, then
\begin{equation}
\label{eq:eigen-phase}
\delta_\ell(k)=-\sum_{j=1}^N\Arg(1+\mu_{\ell,j}(k)).
\end{equation}
Each argument on the right is the principal argument of the individual factor, not of their product.
Indeed $1+t\mu_{\ell,j}(k)$ is a straight nonzero path from $1$.
It cannot cross the negative real axis: a real $\mu_{\ell,j}\le-1$ would contradict the invertibility of $K_{\ell,t}$ at some $t\in(0,1]$.
Thus \eqref{eq:eigen-phase} is exactly \eqref{eq:coupling-log}, including its winding information.

\subsection{Angular-momentum convergence}
\label{sec:convergence}

The following estimate also justifies the differentiation used for time delay.

\begin{lem}
\label{lem:tail}
Let $J=[k_-,k_+]\subset(0,\infty)$.
There are constants $C_J,T_J>0$ and an integer $\ell_J$ such that
\begin{equation}
\label{eq:phase-tail}
\sup_{k\in J}\bigl(|\delta_\ell(k)|+|\delta_\ell'(k)|\bigr)
 \le C_J(1+\ell)^2
       \frac{T_J^{2\ell}}{((2\ell+1)!!)^2},\qquad \ell\ge\ell_J.
\end{equation}
Consequently, the sums of $(2\ell+1)\delta_\ell$ and $(2\ell+1)\delta_\ell'$ converge uniformly on $J$.
The family $\mathsf S(k)-I$ is continuously differentiable in trace norm there.
\end{lem}

\begin{proof}
All arguments $kR_j$ lie in a fixed compact subinterval of $(0,\infty)$.
The power series and derivative identities for spherical Bessel functions \cite[Secs.~10.51--10.53]{DLMF} give, for large $\ell$, $$ |j_\ell(t)|\le C_T\frac{T^\ell}{(2\ell+1)!!},\qquad |j_\ell(t)y_\ell(s)|\le \frac{C_T}{(2\ell+1)s}\left(\frac ts\right)^\ell, $$ for $0<a\le t\le s\le T$.
Differentiation on this compact interval costs at most a polynomial factor in $\ell$.
Since $$ \Im(m_\ell(k^2+i0))_{ij}=k j_\ell(kR_i)j_\ell(kR_j), $$ writing $K_{\ell,t}=A_{\ell,t}+iB_{\ell,t}$ with entrywise real and imaginary parts gives, uniformly for $t\in[0,1]$ and $k\in J$, $$ \|A_{\ell,t}-I_N\|\le\frac{C_J}{2\ell+1},\qquad \|B_{\ell,t}\|\le C_J\frac{T_J^{2\ell}}{((2\ell+1)!!)^2}.
$$ In particular $A_{\ell,t}$ and $K_{\ell,t}$ have uniformly bounded inverses for large $\ell$, and $\det A_{\ell,t}>0$.
Expand $\log\det(I_N+iA_{\ell,t}^{-1}B_{\ell,t})$ near the identity.
Its imaginary part and its $k$ derivative satisfy \eqref{eq:phase-tail}; the real factor $\det A_{\ell,t}$ does not contribute to the phase.
The double factorial dominates every polynomial factor.
Finally, $|S_\ell-1|\le2|\delta_\ell|$ and $|S_\ell'|=2|\delta_\ell'|$, which prove the trace-norm claims after summing with multiplicities.
The finitely many remaining channels are smooth because their determinants do not vanish.
\end{proof}

\begin{rem}\label{rem:phase-branches}
A phase obtained from a scattering coefficient alone is defined modulo $\pi$.
The coupling-normalized phase may differ from such a choice by an integer multiple of $\pi$.
This does not change the scattering length or the threshold expansion of $S_0(k)$, but it is essential for counting bound states in the Friedel sum.
\end{rem}

\subsection{Spectral shift and Friedel excess}
\label{sec:friedel}

The ordinary resolvent difference is trace class by Theorem~\ref{thm:input}.
Behrndt, Gesztesy, and Nakamura expressed the spectral shift for hypersurface interactions in terms of boundary operator logarithms \cite[Theorem~4.3 and Corollary~4.4]{BehrndtGesztesyNakamura2018}.
Below we obtain the corresponding channel sum for concentric shells.
We give a direct normalization argument because the determinant of the scattering matrix alone does not determine the integer part of the phase.

Write $H_\alpha=H$ to emphasize the dependence on the strength vector $\alpha=(\alpha_1,\ldots,\alpha_N)$.
We use the normalized spectral shift function $\xi(E)=\xi(E;H_\alpha,H_0)$ with the convention
\begin{equation}
\label{eq:trace-formula}
\tr\{f(H_\alpha)-f(H_0)\}=\int_\R f'(E)\xi(E)\,dE,
\qquad f\in C_0^\infty(\R).
\end{equation}
It belongs to $L^1(\R;(1+E^2)^{-1}dE)$ and is zero below both spectra.
At a negative energy outside the eigenvalues, $\xi(E)$ is minus the number of eigenvalues of $H_\alpha$ below $E$, including multiplicity.
The Birman-Kre\u{\i}n formula has the corresponding sign \cite{BirmanKrein1962,BirmanYafaev1992,Yafaev1992}:
\begin{equation}
\label{eq:BK}
\det \mathsf S(k)=e^{-2\pi i\xi(k^2)}.
\end{equation}
Lemma~\ref{lem:tail} justifies the Fredholm determinant and its partial-wave product.

Let $E_F>0$ be the Fermi energy.
For a smooth approximation $f_\varepsilon$ of the Fermi occupation step, with its lower cutoff below both spectra, define
\begin{equation}
\label{eq:excess-definition}
\Delta N(E_F)=\lim_{\varepsilon\downarrow0}
 \tr\{f_\varepsilon(H_\alpha)-f_\varepsilon(H_0)\}
 =-\xi(E_F).
\end{equation}
The limit first holds at Lebesgue points, and the positive-energy representative obtained below is continuous.
This is a spectrally regularized excess number.
We do not assume that the difference of the sharp infinite-volume spectral projections is trace class, nor prove here an additional spatial-cutoff version of the sum rule.
For that distinction in potential scattering, see \cite{KohmotoKomaNakamura2013}.

\begin{thm}
\label{thm:Friedel}
With the phases fixed by \eqref{eq:coupling-log}, the normalized positive-energy spectral shift has the continuous representative
\begin{equation}
\label{eq:SSF-real}
\xi(k^2)=-\frac1\pi\sum_{\ell=0}^\infty(2\ell+1)\delta_\ell(k).
\end{equation}
Thus the Friedel sum \cite{Friedel1952} is
\begin{equation}
\label{eq:Friedel}
\Delta N(E_F)=\frac1\pi\sum_{\ell=0}^\infty(2\ell+1)\delta_\ell(k_F),
\qquad k_F=\sqrt{E_F}.
\end{equation}
Equivalently, each $\delta_\ell$ is $-\Im\mathcal L_\ell$.
Both phase sums converge locally uniformly for $k_F>0$.
The real parts of the unregularized logarithms need not be summable.
\end{thm}

\begin{proof}
Combining \eqref{eq:Sell-det} with \eqref{eq:BK} gives the exponential form of \eqref{eq:SSF-real}.
That step alone fixes $\xi$ only modulo integers.
Appendix~\ref{app:normalization} identifies the coupling logarithms with the normalized resolvent perturbation determinant and therefore fixes the integer as well.
Lemma~\ref{lem:tail} gives local uniform convergence.
Equation~\eqref{eq:excess-definition} then proves \eqref{eq:Friedel}.
\end{proof}

The bound-state contribution is already contained in these phases and must not be added again to \eqref{eq:Friedel}.
If every level has an independent internal degeneracy $g_{\rm int}$, the excess number is multiplied by $g_{\rm int}$.
For electrons its electric charge is $-eg_{\rm int}\Delta N$ when $\Delta N$ denotes the spinless count and $e>0$.
This noninteracting model does not impose a self-consistent screening condition.

\subsection{Density-of-states shift and the time-delay trace}
\label{sec:time-delay}

Define the excess density of states at positive energy by $\Delta\rho(E)=d\Delta N(E)/dE$.
In the units used here, the Wigner-Smith operator and its trace are
\begin{equation}
\label{eq:WS}
Q(E)=-i\mathsf S(\sqrt E)^*\frac{d}{dE}\mathsf S(\sqrt E),
\qquad \tau_E(E)=\tr Q(E).
\end{equation}
Lemma~\ref{lem:tail} makes these derivatives and traces well defined on compact positive-energy intervals.
The channel delay is $\tau_{\ell,E}(k^2)=\delta_\ell'(k)/k$.
Differentiating the determinant gives the finite-matrix formula
\begin{align}
\label{eq:delay-matrix}
\tau_E(k^2) &=-\frac1k\Im\sum_{\ell=0}^\infty(2\ell+1)
 \tr\{K_\ell(k)^{-1}\partial_kK_\ell(k)\}\\
\label{eq:delay-DOS}
 &=\frac1k\sum_{\ell=0}^\infty(2\ell+1)\delta_\ell'(k)
 =2\pi\Delta\rho(k^2)=-2\pi\xi'(k^2).
\end{align}
Here $\xi'$ denotes the energy derivative.
The trace is the sum over scattering channels, not a channel-averaged delay.
A positive value corresponds to a positive phase delay relative to the free reference, whereas negative values are possible.
It should not be identified with a nonnegative dwell time \cite{Texier2016}.

\section{A uniform description of the first threshold}
\label{sec:general-first-threshold}

We now allow all attractive strengths to vary.
Write
\[
 H_g=-\Delta-\sum_{j=1}^N g_j\delta(|x|-R_j),
 \qquad g=(g_1,\ldots,g_N)\in(0,\infty)^N.
\]
The radii are fixed throughout this section.
The phase shifts retain the coupling normalization of Section~\ref{sec:phase-normalization}, and $\Delta N(E;g)=-\xi(E;H_g,H_0)$ denotes the spinless excess.
The leading phase profile and the bound-state threshold exponent are consistent with standard threshold theory \cite{Yafaev1992,KlausSimon1980}.
The role of this section is to provide the uniform estimates needed for the geometric results of Section~\ref{sec:critical-geometry}.
We determine the coefficients by a finite matrix, retain the first correction, and control the full angular-momentum sum on both sides of the threshold.

Put $L_g=\diag(R_1\sqrt{g_1},\ldots,R_N\sqrt{g_N})$ and
\[
 \mathcal B_\ell(k;g)=L_gm_\ell(k^2+i0)L_g,
 \qquad M_\ell(k;g)=I_N-\mathcal B_\ell(k;g).
\]
The matrices $M_\ell$ and $K_\ell$ are similar.
Their determinants, and hence their continued phases, are the same.
At zero energy,
\begin{equation}\label{eq:general-BS-zero}
 (\mathcal B_\ell(0;g))_{ij}
 =\sqrt{g_i g_j}\,
  \frac{R_<^{\ell+1}R_>^{-\ell}}{2\ell+1}.
\end{equation}
Write $B(g)=\mathcal B_0(0;g)$ and let $\rho(g)$ be its largest eigenvalue.
Every entry of $B(g)$ is positive.
Thus $\rho(g)$ is simple and has a real unit eigenvector $v(g)$ with positive entries.
We use this choice of $v$ below.
The matrix $B(g)$ is positive definite, since its kernel $\min(R_i,R_j)$ is the Gram matrix of the functions $\mathbf1_{[0,R_i]}$.

Fix $g_*\in(0,\infty)^N$ such that $\rho(g_*)=1$.
Choose a sufficiently small neighborhood $U$ of $g_*$ whose closure is compact in $(0,\infty)^N$ and satisfies the spectral gaps used below.
After decreasing $U$ if necessary, the following quantities are smooth and well defined:
\begin{equation}\label{eq:threshold-data-1}
\eta(g)=1-\rho(g),\qquad r_j(g)=\sqrt{g_j}R_j,\qquad
\beta(g)=(v(g)^Tr(g))^2.
\end{equation}
Set $P(g)=v(g)v(g)^T$ and $Q(g)=I_N-P(g)$, and define
\begin{equation}\label{eq:threshold-data-2}
T(g)=\left[Q(g)(I_N-B(g))Q(g)\right]^{-1}_{Q(g)\C^N}.
\end{equation}
An expression containing $T$ is understood to act on $Q\C^N$.
For $N=1$ this space is zero dimensional and all such expressions are zero.
Define the real symmetric matrix
\begin{equation}\label{eq:threshold-W}
 W(g)_{ij}=\sqrt{g_i g_j}
 \left(\frac{R_<R_>^2}{2}+\frac{R_<^3}{6}\right)
\end{equation}
and the real numbers
\begin{equation}\label{eq:threshold-h-gamma}
 h(g)=r(g)^TQ(g)T(g)Q(g)r(g),\qquad
 \gamma(g)=v(g)^TW(g)v(g)+\beta(g)h(g).
\end{equation}
Here $\beta>0$, $h\ge0$, and $\gamma>0$ on $U$.

\begin{thm}[Uniform first-threshold formula]
\label{thm:general-threshold}
There are a neighborhood $U$ of $g_*$ and constants $k_0,C>0$ with the following properties.

\begin{itemize}
\item[(i)] If $g\in U$ and $\eta(g)>0$, then $H_g$ has no negative
eigenvalue.
If $\eta(g)=0$, it has no negative eigenvalue and has a simple zero-energy $s$-wave resonance.
If $\eta(g)<0$, it has exactly one negative eigenvalue, which is simple and belongs to the $s$-wave channel.
No $\ell\ge1$ channel has a negative or zero-energy state for any $g\in U$.
\item[(ii)] For every $g\in U$ and $0<k<k_0$,
\begin{equation}\label{eq:uniform-full-threshold}
 \left|\Delta N(k^2;g)
 +\frac1\pi\Arg\{\eta(g)-i\beta(g)k+\gamma(g)k^2\}
 -\frac{k}{\pi}h(g)\right|
 \le C\frac{k^3\bigl(|\eta(g)|+k^2\bigr)}{\eta(g)^2+k^2}.
\end{equation}
The argument is the principal argument; its argument lies strictly in the lower half-plane.
In particular, the error in \eqref{eq:uniform-full-threshold} is uniformly $O(k^2)$ on $U$.
If $q=\eta(g)/(\beta(g)k)$ remains bounded, this gives
\begin{equation}\label{eq:general-crossover-correction}
 \Delta N(k^2;g)
 =-\frac1\pi\Arg(q-i)
 +\frac{k}{\pi}\left\{h(g)
       -\frac{\gamma(g)}{\beta(g)(1+q^2)}\right\}
 +O(k^2),
\end{equation}
uniformly for bounded $q$ and $g\in U$.
\item[(iii)] On the critical surface $\eta(g)=0$, put
\begin{equation}\label{eq:critical-length}
 \Lambda(g)=\frac{v(g)^TW(g)v(g)}{\beta(g)}>0.
\end{equation}
Then
\begin{equation}\label{eq:critical-total-expansion}
 \Delta N(k^2;g)=\frac12-\frac{\Lambda(g)}{\pi}k+O(k^3),
 \qquad
 k\tau_E(k^2;g)=-\Lambda(g)+O(k^2).
\end{equation}
Both estimates are uniform on the critical part of a smaller neighborhood of $g_*$.
The critical $s$-wave effective range is $r_{\rm e}=2\Lambda(g)$, with the convention $k\cot\delta_0(k)=\tfrac12r_{\rm e}k^2+O(k^4)$ at resonance.
\item[(iv)] For $\eta(g)<0$, write the unique negative eigenvalue as
$E_b(g)=-\kappa(g)^2$.
As $\eta(g)\uparrow0$ within $U$,
\begin{equation}\label{eq:bound-state-threshold-law}
 \begin{split}
 \kappa(g)&=-\frac{\eta(g)}{\beta(g)}
       +\frac{\gamma(g)\eta(g)^2}{\beta(g)^3}+O(|\eta(g)|^3),\\
 E_b(g)&=-\frac{\eta(g)^2}{\beta(g)^2}
       +\frac{2\gamma(g)\eta(g)^3}{\beta(g)^4}+O(\eta(g)^4).
 \end{split}
\end{equation}
The error constants are uniform.
Thus the same coefficient $\beta$ fixes the coupling-energy scale in \eqref{eq:general-crossover-correction} and the leading bound-state energy.
\end{itemize}
\end{thm}

\begin{proof}
First consider zero energy.
Entrywise,
\[
 0<\mathcal B_\ell(0;g)
 \le\frac{B(g)}{2\ell+1}.
\]
Both matrices are real symmetric positive definite.
Positivity and monotonicity of the spectral radius of nonnegative matrices therefore give
\begin{equation}\label{eq:general-higher-gap}
 \|\mathcal B_\ell(0;g)\|
 \le\frac{\rho(g)}{2\ell+1},\qquad\ell\ge1.
\end{equation}
Choose $U$ so that $\rho(g)<2$ there and so that the second eigenvalue of $B(g)$ stays below $1$ by a fixed positive amount.
The second condition is omitted when $N=1$.
The Birman-Schwinger principle then proves the negative-eigenvalue assertions in (i).
Indeed, the radial Birman-Schwinger matrices at energy $-\kappa^2$ decrease in the Loewner order with $\kappa>0$ and tend to zero as $\kappa\to\infty$; this follows by compressing the free radial resolvent to shell evaluations.
The decrease is strict on every nonzero evaluation vector.
Thus eigenvalues cross one only downward; their value at zero is \eqref{eq:general-BS-zero}.
Inequality \eqref{eq:general-higher-gap} also excludes higher-channel zero states.

At $\rho(g)=1$, the reduced radial function
\[
 u_*(r)=\sum_{j=1}^N\sqrt{g_j}v_j\min(r,R_j)
\]
obeys the zero-energy jump conditions because $u_*(R_i)=v_i/\sqrt{g_i}$.
It is regular at zero and is the nonzero constant $v^Tr$ outside $R_N$.
Thus $r^{-1}u_*(r)$ is a resonance, not an $L^2$ eigenfunction.
Simplicity follows from simplicity of the Perron eigenvalue.

We next record the estimates needed to sum the higher channels down to zero energy.
Uniformly for $g$ in a relatively compact smaller neighborhood and for $0<k<k_0$, the spherical Bessel series give
\begin{align*}
 \|\Re\mathcal B_\ell(k;g)-\mathcal B_\ell(0;g)\|
   &\le Ck^2,\\
 \|\Im\mathcal B_\ell(k;g)\|
   &\le Ck\frac{(kR_N)^{2\ell}}{((2\ell+1)!!)^2},\qquad\ell\ge1.
\end{align*}
The first derivatives satisfy the corresponding bounds $Ck$ and $C(1+\ell)(kR_N)^{2\ell}/((2\ell+1)!!)^2$.
For the imaginary part, these follow directly from $\Im\mathcal B_\ell=k b_\ell b_\ell^T$, where $(b_\ell)_j=\sqrt{g_j}R_jj_\ell(kR_j)$.
For the real part, multiply the power series for $j_\ell$ and $y_\ell$; the constant term is \eqref{eq:general-BS-zero}, and the remaining even series starts at $k^2$.
The normalized coefficients are bounded uniformly for $\ell\ge1$; see \cite[Secs.~10.51--10.53]{DLMF}.
By \eqref{eq:general-higher-gap}, the real parts of $M_\ell$ have uniformly bounded inverses when $k_0$ is sufficiently small.
The same estimates apply along the coupling path from zero to $g$.
Expanding the logarithm after factoring out this real positive matrix, and summing the double-factorial bounds, gives
\begin{equation}\label{eq:uniform-higher-tail-zero}
 \sum_{\ell\ge1}(2\ell+1)|\delta_\ell(k;g)|\le Ck^3,
 \qquad
 \sum_{\ell\ge1}(2\ell+1)|\partial_k\delta_\ell(k;g)|\le Ck^2.
\end{equation}
The sums are uniform in $g\in U$.

For the $s$ wave the finite matrix is elementary:
\[
 (\mathcal B_0(k;g))_{ij}
 =\sqrt{g_i g_j}\,
   \frac{\sin(kR_<)e^{ikR_>}}{k}.
\]
Its entire continuation in $k$ has the expansion
\begin{equation}\label{eq:general-B-expansion}
 \mathcal B_0(k;g)=B(g)+ik r(g)r(g)^T-k^2W(g)+O(k^3).
\end{equation}
The remainder and its derivatives are uniform in $g\in U$.
Use the orthogonal decomposition $P\C^N\oplus Q\C^N$ and put
\[
 C_Q(k;g)=Q M_0(k;g)Q\big|_{Q\C^N}.
\]
This matrix is invertible for all sufficiently small complex $k$, uniformly on $U$, and $C_Q(0;g)$ is positive definite.
The Schur complement gives the exact factorization
\begin{equation}\label{eq:general-schur-factorization}
 \det M_0(k;g)=\det C_Q(k;g)\,f(k;g),
\end{equation}
where
\[
 f=v^TM_0v-v^TM_0Q\,C_Q^{-1}Q M_0v.
\]
Since $QBv=0$, \eqref{eq:general-B-expansion} yields
\begin{equation}\label{eq:general-schur-expansion}
 f(k;g)=\eta(g)-i\beta(g)k+\gamma(g)k^2+O(k^3).
\end{equation}
The sign of the Schur correction is positive: the two off-diagonal terms have leading product $(-ik)^2\beta\,r^TQTQr$, which is subtracted in the Schur complement.
Also,
\begin{equation}\label{eq:general-stable-phase}
 -\Im\log\det C_Q(k;g)=k h(g)+O(k^3).
\end{equation}
Here the logarithm is continued from its real value at $k=0$.
The absence of a quadratic imaginary term follows from $C_Q(-k;g)=\overline{C_Q(k;g)}$ for real $k$.

For sufficiently small positive $k$, $f(k;g)$ lies in the lower half-plane.
The continuous phase obtained from \eqref{eq:general-schur-factorization} agrees with the canonical phase.
To fix its integer, take a nearby subcritical parameter.
If $\rho(g)<1$, then $I_N-tB(g)$ is positive definite for every $0\le t\le1$.
The coupling-continued logarithm is therefore real at zero energy, and its phase tends to zero as $k\downarrow0$.
Now use continuity in $g$.
There is no determinant zero at positive energy.
Hence
\[
 \delta_0(k;g)=-\Arg f(k;g)-\Im\log\det C_Q(k;g)
\]
throughout a connected sufficiently small $U$.
Since $\beta$ is bounded below by a positive constant,
\[
 |\eta-i\beta k+\gamma k^2|\ge c(|\eta|+k).
\]
Equations \eqref{eq:general-schur-expansion} and \eqref{eq:general-stable-phase} therefore give a phase error bounded by $Ck^3/(|\eta|+k)$.
Combining this with \eqref{eq:uniform-higher-tail-zero} and Theorem~\ref{thm:Friedel} first gives the weaker bound $Ck^3/(|\eta|+k)$.
For the stated sharper estimate, write $f_2=\eta-i\beta k+\gamma k^2$.
Reality symmetry gives $\Re(f-f_2)=O(k^4)$ and $\Im(f-f_2)=O(k^3)$.
Consequently
\[
 \left|\Im\frac{f-f_2}{f_2}\right|
 \le C\frac{k^3(|\eta|+k^2)}{\eta^2+k^2}.
\]
The quadratic remainder in $\log(1+(f-f_2)/f_2)$ is bounded by $Ck^6/(\eta^2+k^2)$ and is absorbed in this estimate for small $k$.
The stable phase and the higher-channel contribution are $O(k^3)$; these are also absorbed since $|\eta|$ stays bounded and small.
This proves \eqref{eq:uniform-full-threshold}.
Expanding the argument after division by $\beta k$ proves \eqref{eq:general-crossover-correction}.

If $\eta=0$, the scalar determinant has a simple zero at $k=0$.
More precisely, $D_0(k)/(-ik)$ has positive value $\det C_Q(0;g)\,\beta(g)$ at $k=0$ and satisfies the usual reality symmetry.
Its analytic logarithm therefore has an odd imaginary part on the real axis.
Thus
\[
 \delta_0(k;g)=\frac\pi2
   +\left(h-\frac\gamma\beta\right)k+O(k^3)
 =\frac\pi2-\frac{v^TWv}{\beta}k+O(k^3).
\]
This proves the first assertion of (iii).
Differentiation, together with the second bound in \eqref{eq:uniform-higher-tail-zero}, proves the delay formula.
Expansion of $k\cot\delta_0(k)$ gives the stated effective range.

Finally, at $k=i\kappa$ the Schur factor is real and
\[
 f(i\kappa;g)=\eta(g)+\beta(g)\kappa
                   -\gamma(g)\kappa^2+O(\kappa^3).
\]
Its derivative with respect to $\kappa$ is positive near zero, uniformly on $U$.
The implicit function theorem gives a unique small root with $\kappa=-\eta/\beta+O(\eta^2)$, which is positive precisely when $\eta<0$.
Part (i) identifies it with the unique bound state.
Substitution back into the scalar equation gives the second term $\gamma\eta^2/\beta^3$ with remainder $O(|\eta|^3)$.
Squaring proves \eqref{eq:bound-state-threshold-law}.
\end{proof}

\subsection{An exact threshold coordinate and a uniform cubic phase remainder}
\label{sec:exact-threshold-coordinate}

Fix a positive first-critical vector $g_*$ and write $R=R_N$.
For the interior radial problem on $(0,R)$ retain the interactions at $R_j$, $j<N$, and impose Dirichlet conditions at $0$ and $R$.
Its lowest eigenvalue is strictly positive at $g_*$ and remains bounded away from zero on a sufficiently small neighborhood $U$.
For $z$ near zero, let $y_z(r;g_{\rm in})$ be the solution regular at zero, satisfying the inner jump conditions, and normalized by $y_z(R;g_{\rm in})=1$.
Here $g_{\rm in}=(g_1,\ldots,g_{N-1})$.
Define
\begin{equation}\label{eq:DN-threshold-data}
 L(z;g_{\rm in})=y_z'(R-;g_{\rm in}),\qquad
 p(g)=g_N-L(0;g_{\rm in}),\qquad
 \mathcal I(g_{\rm in})=\int_0^R y_0(r;g_{\rm in})^2\,dr.
\end{equation}
These are analytic real functions of the strengths near $g_*$, and $L$ is analytic in $z$ near zero.
Unlike a linearized coupling coordinate, $p$ vanishes exactly on the first-critical surface.
The Dirichlet-to-Neumann reduction and effective-range expansion are standard tools; the estimate below concerns their uniform use in the normalized sum over all angular momenta.

\begin{thm}[Uniform phase normal form in an exact coupling coordinate]
\label{thm:DN-normal-form}\label{thm:exact-normal-form}
After decreasing $U$, there are $k_0,C>0$ such that, for all $g\in U$ and $0<k<k_0$,
\begin{equation}\label{eq:DN-phase-normal-form}
 \left|\Delta N(k^2;g)
 -\left\{\frac12+\frac1\pi
       \arctan\left(\frac{p(g)}k+\mathcal I(g_{\rm in})k\right)
                  -\frac{Rk}\pi\right\}\right|\le Ck^3.
\end{equation}
The arctangent takes values in $(-\pi/2,\pi/2)$.
The differentiated formula also has a uniform remainder:
\begin{equation}\label{eq:DN-delay-normal-form}
 \left|k\tau_E(k^2;g)
 -\left\{-R+
       \frac{\mathcal I(g_{\rm in})k^2-p(g)}
            {k^2+(p(g)+\mathcal I(g_{\rm in})k^2)^2}\right\}\right|
 \le Ck^2.
\end{equation}
The regimes $p<0$, $p=0$, and $p>0$ are, respectively, subcritical, critical, and supercritical.
The unique bound state on $p>0$ satisfies
\begin{equation}\label{eq:DN-pole-expansion}
 \kappa(g)=p(g)-\mathcal I(g_{\rm in})p(g)^2+O(p(g)^3),\qquad
 E_b(g)=-p(g)^2+2\mathcal I(g_{\rm in})p(g)^3+O(p(g)^4),
\end{equation}
uniformly near the critical surface.
For $p\ne0$ the scattering length is exactly
\begin{equation}\label{eq:DN-scattering-length}
 a(g)=R+\frac1{p(g)}.
\end{equation}
On $p=0$ one has $y_0=\phi$, the exterior-normalized critical resonance, and
\begin{equation}\label{eq:DN-critical-range}
 r_{\rm e}=2(R-\mathcal I).
\end{equation}
The matrix coordinate and the exact coordinate are related by $p(g)=-\eta(g)/\beta(g)+O(\eta(g)^2)$, uniformly in a smaller neighborhood of the critical surface.
\end{thm}

\begin{proof}
At $g_*$ the normalized critical resonance $\phi$ is positive on $(0,R]$.
For an interior Dirichlet test function, the ground-state identity gives
\[
 \int_0^R|u'|^2\,dr-\sum_{j<N}(g_*)_j|u(R_j)|^2
 =\int_0^R\phi^2\left|\left(\frac{u}{\phi}\right)'\right|^2\,dr.
\]
Thus the interior operator is nonnegative.
Its kernel is zero, because equality would make $u$ proportional to $\phi$, whereas $u(R)=0$ and $\phi(R)=1$.
Compact resolvent then gives a positive lowest eigenvalue.
Standard finite-interval ODE dependence, or the finite transfer matrices between shells, supplies the asserted uniform gap and analytic dependence near $g_*$.

For real $z$, differentiate the equation for $y_z$ in $z$.
The Wronskian $y_z\partial_z y_z'-y_z'\partial_z y_z$ has derivative $-y_z^2$ between the shells and is continuous across each shell.
The normalization gives $\partial_z y_z(R)=0$, while the Wronskian vanishes at zero.
Consequently
\[
 \partial_zL(z;g_{\rm in})=-\int_0^R y_z(r;g_{\rm in})^2\,dr,
\]
and, with uniform derivative bounds,
\begin{equation}\label{eq:DN-even-expansion}
 L(k^2;g_{\rm in})-g_N=-p(g)-\mathcal I(g_{\rm in})k^2+O(k^4).
\end{equation}

The real exterior reduced scattering solution is proportional to $\sin(kr+\delta_0(k;g))$.
The outer jump condition gives exactly
\[
 k\cot(kR+\delta_0(k;g))=L(k^2;g_{\rm in})-g_N.
\]
It follows that, for the coupling-normalized phase,
\begin{equation}\label{eq:DN-exact-phase}
 \begin{aligned}
 \delta_0(k;g)&=\frac\pi2+
       \arctan\left(\frac{p(g)}k+\mathcal I(g_{\rm in})k+r(k,g)\right)-Rk,\\
 |r(k,g)|&\le Ck^3,\qquad |\partial_kr(k,g)|\le Ck^2.
 \end{aligned}
\end{equation}
To fix the integer, take a subcritical nearby coupling and send $k$ to zero: both sides tend to zero.
Continuity in $g$, and the absence of positive-energy determinant zeros, fix the same integer throughout a connected smaller $U$.
Since the arctangent is Lipschitz with constant one, \eqref{eq:DN-exact-phase} has error $O(k^3)$ after deleting $r$.
The uniform higher-channel bounds \eqref{eq:uniform-higher-tail-zero} prove \eqref{eq:DN-phase-normal-form}.

For the derivative, fix $g$ and put
\[
 x=\frac{p(g)}{k}+\mathcal I(g_{\rm in})k,
 \qquad x'=-\frac{x}{k}+2\mathcal I(g_{\rm in}).
\]
The remainder in \eqref{eq:DN-exact-phase} is obtained by dividing the $O(k^4)$ term in the Taylor expansion of $L(k^2)$ by $k$.
Its analytic dependence on $k^2$ gives $|r|\le Ck^3$ and $|r'|\le Ck^2$, uniformly on the coupling neighborhood.
For $k$ small enough that $|r|\le 1/2$, elementary denominator comparisons give
\begin{equation}\label{eq:arctan-derivative-comparison}
 \left|\frac{1}{1+(x+r)^2}-\frac{1}{1+x^2}\right|
 \le C|r|\frac{1+|x|}{(1+x^2)^2},
 \qquad
 \frac{1}{1+(x+r)^2}\le\frac{C}{1+x^2}.
\end{equation}
Consequently,
\begin{align*}
 &\left|\partial_k\arctan(x+r)-\partial_k\arctan x\right|\\
 &\quad\le\frac{|r'|}{1+(x+r)^2}
   +|x'|\left|\frac{1}{1+(x+r)^2}-\frac{1}{1+x^2}\right|\\
 &\quad\le Ck^2+Ck^3\left(\frac{|x|}{k}+2\mathcal I\right)
                 \frac{1+|x|}{(1+x^2)^2}
 \le C'k^2.
\end{align*}
The last inequality holds for all real $x$ since
$|x|(1+|x|)/(1+x^2)^2$ is bounded.
In particular, no condition $|p(g)|\le Ck$ or $|p(g)|\ge ck$ is used.
This is the point required for a uniform estimate on a full neighborhood of the critical surface.
Since
\[
 \partial_k\arctan(p/k+\mathcal I k)
       =\frac{\mathcal I k^2-p}{k^2+(p+\mathcal I k^2)^2},
\]
the differentiated higher-channel estimate proves \eqref{eq:DN-delay-normal-form}.
This argument controls the derivative directly and does not differentiate an unspecified big-$O$ term.

At zero energy the outer matching condition is exactly $p=0$.
The positive normalized solution remains positive near $g_*$, so this condition describes the first-critical surface.
Monotonicity in the outer attractive strength identifies the two signs of $p$.
For a bound state, the exterior logarithmic derivative is $-\kappa$; thus
\[
 0=\kappa+L(-\kappa^2;g_{\rm in})-g_N
   =\kappa-p+\mathcal I\kappa^2+O(\kappa^4).
\]
The implicit function theorem proves \eqref{eq:DN-pole-expansion}.
The same equation defines a signed analytic root across $p=0$.
Comparing this root with the Schur expansion \eqref{eq:general-schur-expansion} proves the stated relation to $\eta$ and $\beta$, including its uniform remainder.

For $p\ne0$, the exterior zero-energy solution is proportional to $r-a$ and its logarithmic derivative at $R$ is $(R-a)^{-1}$.
The outer jump condition makes this derivative $-p$, proving \eqref{eq:DN-scattering-length}.
Finally, on $p=0$, \eqref{eq:DN-exact-phase} gives $\delta_0=\pi/2-(R-\mathcal I)k+O(k^3)$, hence $k\cot\delta_0=(R-\mathcal I)k^2+O(k^4)$ and \eqref{eq:DN-critical-range} follows.
\end{proof}

\begin{rem}[Evaluation of the exact coordinate]
Put $G_{ij}=\min(R_i,R_j)$ and partition
\[
 G^{-1}-\diag(g)=
 \begin{pmatrix}A&b\\b^T&d\end{pmatrix}
\]
into the first $N-1$ indices and the last index.
Near $g_*$ the matrix $A$ is positive definite.
Indeed, at $g_*$ the full matrix is positive semidefinite with kernel spanned by the resonance values, whose last component is one.
Its restriction to vectors with last component zero is therefore positive definite, and this property persists in a neighborhood.
Its entries and $b$ do not depend on $g_N$.
Discrete minimization of the zero-energy radial form with last value one gives
\[
 L(0)=(G^{-1})_{NN}-b^TA^{-1}b,
 \qquad
 p=-\frac{\det(I-B(g))}{\det G\,\det A}.
\]
For $N=1$, the complementary matrix is empty, its determinant is one, and the formula reads $p=g_1-1/R_1$.
The shell values of $y_0$ are $(-A^{-1}b,1)$; linear interpolation therefore evaluates $\mathcal I$ by elementary integration.
Indeed, $u^TG^{-1}u$ is the Dirichlet energy of the piecewise linear interpolant with values $u_j$ at $R_j$ and value zero at the origin.
Its constrained minimum after the shell terms are subtracted is $d-b^TA^{-1}b=L(0)-g_N=-p$.
The determinant identity is the corresponding Schur factorization.
\end{rem}

\subsection{The effective range and transverse paths}

\begin{cor}[Geometric bounds at the first resonance]
\label{cor:critical-geometric-bounds}
At any point of the critical surface $\rho(g)=1$,
\begin{equation}\label{eq:geometric-effective-range}
 \frac{2R_1}{3}\le\Lambda(g)\le\frac{2R_N}{3},
 \qquad
 \frac{4R_1}{3}\le r_{\rm e}(g)\le\frac{4R_N}{3}.
\end{equation}
For $N\ge2$ and strictly positive strengths both inequalities are strict.
The endpoints are sharp as limiting values when all but the innermost or outermost shell strengths tend to zero while the remaining strength is adjusted to keep the first resonance.
In particular, the critical time delay is negative for all sufficiently small positive energies.
\end{cor}

\begin{proof}
Set
\[
 w_j=\frac{v_j\sqrt{g_j}R_j}{v^Tr},\qquad
 \sum_jw_j=1,\qquad w_j>0.
\]
Substitution into \eqref{eq:critical-length} gives
\begin{equation}\label{eq:critical-weight-formula}
 \Lambda(g)=\sum_{i,j=1}^Nw_iw_j
 \left(\frac{R_>}{2}+\frac{R_<^2}{6R_>}\right).
\end{equation}
For $R_1\le x\le y\le R_N$, the expression $y/2+x^2/(6y)$ is increasing in both variables and lies between $2R_1/3$ and $2R_N/3$.
The weights in \eqref{eq:critical-weight-formula} are positive and sum to one, which proves the bounds and strictness.
The single-shell limits give sharpness.

There is also a zero-energy interpretation.
Normalize the reduced resonance by its exterior value and write
\[
 \phi_*(r)=\frac{u_*(r)}{v^Tr}
          =\sum_jw_j\min\{r/R_j,1\}.
\]
A direct integration gives
\begin{equation}\label{eq:critical-Bethe-integral}
 \Lambda(g)=\int_0^{R_N}\{1-\phi_*(r)^2\}\,dr.
\end{equation}
This is the zero-energy effective-range integral in the present piecewise linear model.
The positive weights yield the more specific bounds \eqref{eq:geometric-effective-range}.
\end{proof}

The integral in \eqref{eq:critical-Bethe-integral} is Bethe's effective-range formula at infinite scattering length \cite{Bethe1949}; see also \cite{HammerLee2010}.
For a larger class of finite-range interactions the causality bound at this point is $r_{\rm e}\le2R_N$.
The sharper interval \eqref{eq:geometric-effective-range} uses the positive, nondecreasing, concave first-threshold radial solution.
The same comparison applies to purely attractive radial potentials with an empty core of radius $R_1$ and support inside $R_N$.
Its endpoints and positivity are consequences of that additional structure.
Formula \eqref{eq:critical-weight-formula} makes the dependence on all shell radii and strengths explicit.

\begin{cor}[A transverse coupling path]
\label{cor:transverse-coupling}
Let $s\mapsto g(s)$ be a $C^2$ path of positive strengths, with $g(s_*)=g_*$ and $a_*:=\frac{d}{ds}\rho(g(s_*))>0$.
Put $\beta_*=\beta(g_*)$.
For $s>s_*$ sufficiently close to $s_*$, the unique bound state satisfies
\[
 E_b(g(s))=-\left(\frac{a_*}{\beta_*}\right)^2(s-s_*)^2
           +O(|s-s_*|^3).
\]
For fixed bounded $q$, taking $s=s_*-(\beta_*/a_*)qk$ gives
\[
 \Delta N(k^2;g(s))=-\frac1\pi\Arg(q-i)+O(k)
\]
uniformly as $k\downarrow0$.
For common scaling $g(s)=s\widehat g$, one has $s_*=1/\rho(\widehat g)$ and $a_*=\rho(\widehat g)$.
\end{cor}

\begin{proof}
The simplicity of $\rho$ gives $\eta(g(s))=-a_*(s-s_*)+O((s-s_*)^2)$, while $\beta(g(s))=\beta_*+O(|s-s_*|)$.
Apply parts (ii) and (iv) of Theorem~\ref{thm:general-threshold}.
\end{proof}

\begin{rem}
If $N\ge2$, every individual shell is strictly subcritical at the first resonance: $g_jR_j<1$.
Indeed, the $j$th component of $Bv=v$ gives $g_jR_jv_j<v_j$ because all off-diagonal entries and all components of $v$ are positive.
The same inequalities remain valid in a sufficiently small neighborhood of the critical configuration.
Collective binding itself is not claimed as a new mechanism; the formulas above quantify its effect on the normalized total spectral shift and the emerging bound-state energy.
\end{rem}

\section{Geometry and inverse description of the first threshold}
\label{sec:critical-geometry}

\subsection{Geometry of the critical surface and tangent binding}
\label{sec:tangent-binding}

The quadratic bound-state law concerns a transverse change of coupling.
Klaus studied perturbations of a resonant operator \cite{Klaus1982}.
Golovaty obtained two-term threshold expansions, including a fourth-order binding law when the first resonant perturbation term vanishes \cite[Theorem~2]{Golovaty2020}.
For a one-dimensional potential model, a related two-sided binding result for small nonzero perturbations is given in \cite[Corollary~3]{Golovaty2020}.
Here the parameter dependence of $\rho$ is a largest-eigenvalue problem for a symmetric matrix affine in $g$.
Its convexity gives only nonnegative second derivatives.
The distinctive point below is that every nonzero tangent direction has a strictly positive second derivative that can be calculated without diagonalizing a perturbed matrix.
The coefficient is the energy of a weighted path graph and also determines the emerging bound-state pole and the total spectral-shift crossover.
Thus the statement is not presented as a new fourth-order exponent.

\begin{thm}[Strict curvature and quartic binding]
\label{thm:tangent-binding}
Let $N\ge2$, let $g_*>0$ satisfy $\rho(g_*)=1$, and normalize its reduced zero-energy resonance by $\phi(R_N)=1$.
For a vector $a\in\R^N\setminus\{0\}$ assume the tangency condition
\begin{equation}\label{eq:tangent-condition}
 \sum_{j=1}^N a_j\phi(R_j)^2=0.
\end{equation}
Set $g(\varepsilon)=g_*+\varepsilon a$, and define
\begin{equation}\label{eq:tangent-C}
 C_a=\sum_{j=1}^{N-1}
       \frac{R_{j+1}-R_j}{\phi(R_j)\phi(R_{j+1})}
       \left(\sum_{i=1}^j a_i\phi(R_i)^2\right)^2.
\end{equation}
Then $C_a>0$.
Let $\mathcal C=\{g>0:\rho(g)=1\}$ and equip it with the outward unit normal
$\nu(g)=\nabla\rho(g)/\|\nabla\rho(g)\|$.
With the convention
$\mathrm{II}^{\mathrm{out}}_{g_*}(a,a)=\langle D_a\nu,a\rangle$, one has
\begin{equation}\label{eq:tangent-second-fundamental}
 D^2\rho(g_*)[a,a]=2\beta_* C_a,
 \qquad
 \mathrm{II}^{\mathrm{out}}_{g_*}(a,a)
    =\frac{2\beta_*C_a}{\|\nabla\rho(g_*)\|}
    =\frac{2C_a}{\bigl(\sum_{j=1}^N\phi(R_j)^4\bigr)^{1/2}}>0.
\end{equation}
Thus $\mathcal C$ has a positive definite second fundamental form at every positive first-critical coupling.
For a tangent vector $a$ with first $N-1$ components $a_{\rm in}$, the critical surface written as a graph $g_N=g_N^{\rm crit}(g_{\rm in})$ satisfies
\begin{equation}\label{eq:tangent-graph}
 g_N^{\rm crit}((g_*)_{\rm in}+\varepsilon a_{\rm in})
  =(g_*)_N+\varepsilon a_N-C_a\varepsilon^2+O(|\varepsilon|^3).
\end{equation}
For all sufficiently small nonzero $\varepsilon$, of either sign, $H_{g(\varepsilon)}$ has exactly one negative eigenvalue, and
\begin{align}
 \rho(g(\varepsilon))
   &=1+\beta_* C_a\varepsilon^2+O(|\varepsilon|^3),
       \label{eq:tangent-rho}\\
 \kappa(g(\varepsilon))
   &=C_a\varepsilon^2+O(|\varepsilon|^3),\qquad
 E_b(g(\varepsilon))=-C_a^2\varepsilon^4+O(|\varepsilon|^5).
       \label{eq:tangent-energy}
\end{align}
Moreover, uniformly for sufficiently small real $\varepsilon$ and $0<k<k_0$,
\begin{equation}\label{eq:tangent-SSF}
 \Delta N(k^2;g(\varepsilon))
 =\frac12+\frac1\pi\arctan\frac{C_a\varepsilon^2}{k}
      +O(|\varepsilon|+k).
\end{equation}
Thus the crossover occurs on the scale $k\asymp\varepsilon^2$.
In a joint limit $k,|\varepsilon|\downarrow0$, the excess tends to $1$ if $k/\varepsilon^2\to0$, to $1/2$ if $k/\varepsilon^2\to\infty$, and to $1/2+\pi^{-1}\arctan(C_a/q)$ if $k/\varepsilon^2\to q\in(0,\infty)$.
\end{thm}

\begin{proof}
Write $G_{ij}=\min(R_i,R_j)$, $\Phi=(\phi(R_j))_{j=1}^N$, and $T=G^{-1}-\diag(g_*)$.
The critical equation is $T\Phi=0$.
Since $G^{1/2}TG^{1/2}=I-A_*$, where $A_*=G^{1/2}\diag(g_*)G^{1/2}$ has the same spectrum as $B(g_*)$, $T$ is positive semidefinite with this one-dimensional kernel.
Let $T^+$ be its inverse on $\Phi^\perp$, extended by zero on the kernel, and put $q=\diag(a)\Phi$.
Tangency is exactly $q\perp\Phi$.

For the negative-energy free radial Green matrix one has
\[
 (G_\kappa)_{ij}
 =\frac{e^{-\kappa|R_i-R_j|}-e^{-\kappa(R_i+R_j)}}{2\kappa}
 =G_{ij}-\kappa R_iR_j+O(\kappa^2).
\]
Because $G^{-1}(R_1,\ldots,R_N)^T=e_N$, its inverse satisfies
\[
 G_\kappa^{-1}-\diag(g_*+\varepsilon a)
 =T-\varepsilon\diag(a)+\kappa e_Ne_N^T+O(\kappa^2).
\]
Use the Schur complement relative to $\operatorname{span}\{\Phi\}\oplus\Phi^\perp$, multiplied by $\|\Phi\|^2$.
Since $\Phi_N=1$, its scalar zero condition is
\[
 0=\kappa-\varepsilon\sum_j a_j\Phi_j^2
        -\varepsilon^2q^TT^+q
        +O(\kappa^2+|\kappa\varepsilon|+|\varepsilon|^3).
\]
The implicit function theorem and tangency give $\kappa=\varepsilon^2q^TT^+q+O(|\varepsilon|^3)$.

It remains to evaluate this positive quadratic form.
The inverse of $G$ is tridiagonal, with off-diagonal entries $-1/(R_{j+1}-R_j)$.
Consequently the matrix $\diag(\Phi)T\diag(\Phi)$ is the path-graph Laplacian with conductances $c_j=\Phi_j\Phi_{j+1}/(R_{j+1}-R_j)$: its row sums are zero because $T\Phi=0$.
If $\psi=T^+q$ and $y_j=\psi_j/\Phi_j$, then this Laplacian sends $y$ to $p=(a_j\Phi_j^2)_j$.
The current across its $j$th edge is therefore $c_j(y_j-y_{j+1})=\sum_{i\le j}p_i$.
Its energy identity gives
\[
 q^TT^+q=p^Ty
 =\sum_{j=1}^{N-1}\frac{1}{c_j}
       \left(\sum_{i\le j}p_i\right)^2=C_a.
\]
All conductances are positive.
If $C_a=0$, every partial sum of $p$ vanishes; tangency gives its total sum zero, so $p=0$ and $a=0$, a contradiction.
Thus the small root is positive for both signs of nonzero $\varepsilon$.
Theorem~\ref{thm:general-threshold}(i) and continuity of the remaining Birman-Schwinger eigenvalues identify it with the unique negative eigenvalue.
Squaring proves \eqref{eq:tangent-energy}.

For completeness, $\rho(g)$ is the largest eigenvalue of $G^{1/2}\diag(g)G^{1/2}$, which depends linearly on $g$.
The largest-eigenvalue characterization makes $\rho$ convex.
However, convexity itself supplies only a nonnegative Hessian; strict positivity and the coefficient in \eqref{eq:tangent-C} follow from the weighted path-graph argument above.
To compute its gradient, let $v_*$ be the unit Perron vector of $B(g_*)$ and put $z_*=G^{1/2}\diag(\sqrt{g_*})v_*$.
Then $\|z_*\|=1$, $A_*z_*=z_*$, and $G^{1/2}z_*=\sqrt{\beta_*}\Phi$.
Differentiating the simple eigenvalue gives
\begin{equation}\label{eq:critical-gradient}
 \partial_{g_j}\rho(g_*)=\beta_*\Phi_j^2.
\end{equation}
Thus the assumed condition is precisely tangency to its level set.
The bound-state law in Theorem~\ref{thm:general-threshold}(iv), together with $\beta(g_*+\varepsilon a)=\beta_*+O(|\varepsilon|)$, gives \eqref{eq:tangent-rho}.
Differentiation along the line gives $D^2\rho(g_*)[a,a]=2\beta_*C_a$.
Since $\nabla\rho(g_*)=\beta_*(\Phi_j^2)_{j=1}^N$, the definition of the outward second fundamental form gives \eqref{eq:tangent-second-fundamental}.
The indicated sign convention is that of the outward normal of $\{g:\rho(g)\le1\}$.

The exact coordinate of Theorem~\ref{thm:exact-normal-form} satisfies $p=-\eta/\beta+O(\eta^2)$ near the critical surface.
At $g_*$ this implies
\begin{equation}\label{eq:tangent-p-derivatives}
 Dp(g_*)[b]=\sum_{j=1}^N b_j\Phi_j^2
 \quad(b\in\R^N),
 \qquad D^2p(g_*)[a,a]=2C_a
 \quad(a\in T_{g_*}\mathcal C).
\end{equation}
Indeed, for the tangent straight line, the pole equation in \eqref{eq:DN-pole-expansion}, together with \eqref{eq:tangent-energy}, gives $p(g_*+\varepsilon a)=C_a\varepsilon^2+O(|\varepsilon|^3)$.
Here the pole and $p$ have the same leading quadratic term, because $\kappa=p+O(p^2)$.
Since $p(g)=g_N-L(0;g_{\rm in})$, the identity $p=0$ describes the critical graph and the last expansion gives \eqref{eq:tangent-graph}.
This also shows directly that the graph bends toward smaller outer-shell strength in every nonzero tangent direction.

Finally, the complex argument in \eqref{eq:uniform-full-threshold} differs from $\beta_*(-C_a\varepsilon^2-ik)$ by $O(|\varepsilon|^3+|\varepsilon|k+k^2)$, whereas the modulus of the latter is bounded below by a positive constant times $\varepsilon^2+k$.
The difference of their arguments is consequently $O(|\varepsilon|+k)$.
The other terms in the uniform formula have the same or smaller order.
Since $-\Arg(-x-i)=\pi/2+\arctan x$ for $x\ge0$, this proves \eqref{eq:tangent-SSF} and its limits.
\end{proof}

\begin{rem}
The fourth-order threshold mechanism was already studied in \cite{Klaus1982,Golovaty2020}; the exponent is not claimed as new.
The specific content of Theorem~\ref{thm:tangent-binding} is the strictly positive explicit coefficient \eqref{eq:tangent-C} for \emph{every} nonzero straight tangent redistribution, its equality with the outward curvature coefficient, and its appearance in the normalized total spectral-shift crossover.
The strictness follows from the connected weighted path graph rather than from general convexity alone.
A curved path can instead stay on the critical surface.
In fact, formula~\eqref{eq:tangent-graph} gives the exact quadratic correction required to remain critical.
Thus the two-sided binding conclusion is specifically a statement about straight paths.
\end{rem}

\begin{rem}[Curved paths]
For $g(\varepsilon)=g_*+\varepsilon a+\varepsilon^2b$ with $\sum a_j\phi_j^2=0$, Taylor expansion of the exact coordinate gives
\[
 p(g(\varepsilon))=
 \left(C_a+\sum_jb_j\phi_j^2\right)\varepsilon^2+O(|\varepsilon|^3).
\]
Here $\partial_{g_j}p(g_*)=\phi_j^2$, by the relation to $\eta/\beta$ and \eqref{eq:critical-gradient}.
A positive coefficient produces binding on both sides; a negative coefficient gives no negative eigenvalue nearby.
If it vanishes, higher-order terms decide the result.
In that case the exact $p$ in \eqref{eq:DN-phase-normal-form} must be retained when $k$ is very small.
Thus straight tangent lines and curves contained in the critical surface have different spectral behavior.
\end{rem}

\begin{rem}[Two-shell example]
Take $R_1=1$, $R_2=3$, and $g_*=(1/2,1/4)$.
The normalized resonance has values $\phi(R_1)=1/2$ and $\phi(R_2)=1$.
Thus $a=(1,-1/4)$ is tangent.
On $(1,3)$ one has $\phi(r)=(r+1)/4$, so $C_a=(3-1)(1/4)^2/((1/2)\cdot1)=1/4$.
The critical graph can be found directly from the zero-energy determinant:
\[
 g_2^{\rm crit}(g_1)=\frac{1-g_1}{3-2g_1},
 \qquad
 g_2^{\rm crit}(1/2+\varepsilon)
 =\frac14-\frac{\varepsilon}{4}-\frac{\varepsilon^2}{4}
  +O(|\varepsilon|^3).
\]
Its quadratic term is precisely $-C_a$, as required by \eqref{eq:tangent-graph}.
Consequently
\[
 E_b\bigl(1/2+\varepsilon,1/4-\varepsilon/4\bigr)
       =-\varepsilon^4/16+O(|\varepsilon|^5).
\]
\end{rem}

\subsection{An inverse parametrization of the first-critical strengths}
\label{sec:inverse-critical}

The previous subsection describes the second-order response of a critical coupling to an infinitesimal redistribution.
We now give global coordinates on the same hypersurface.
They determine which effective-range responses are possible at fixed radii, and how the strengths can be recovered from threshold measurements.
This is an inverse classification, not only an estimate for an individual resonance.
The general effective-range integral is due to Bethe \cite{Bethe1949}; our result is its exact image and inverse parametrization within the positive finite-shell family.

\begin{thm}[Critical configurations and probability weights]
\label{thm:inverse-critical}
For fixed distinct radii, let $\mathcal W=\{w\in(0,\infty)^N:\sum_jw_j=1\}$.
For $w\in\mathcal W$ put
\begin{equation}\label{eq:inverse-critical-map}
 \phi_w(r)=\sum_jw_j\min\{r/R_j,1\},\qquad
 g_j(w)=\frac{w_j}{R_j\phi_w(R_j)}.
\end{equation}
This is a smooth bijection from $\mathcal W$ onto $\{g\in(0,\infty)^N:\rho(g)=1\}$, with smooth inverse.
The map extends continuously and bijectively to the closed simplex and the first-critical configurations with nonnegative strengths, where inactive shells are omitted.
In these coordinates,
\begin{align}
 \beta(w)^{-1}
   &=\int_0^{R_N}\phi_w'(r)^2\,dr
     =\sum_{i,j}\frac{w_iw_j}{\max(R_i,R_j)},
          \label{eq:beta-inverse-weights}\\
 r_{\rm e}(w)
   &=\frac43\sum_jw_jR_j
     +\frac13\sum_{i,j}w_iw_j
        \frac{(R_i-R_j)^2}{\max(R_i,R_j)}.
          \label{eq:range-variance-weights}
\end{align}
For $N\ge2$, $r_{\rm e}$ is strictly concave on the simplex, and its image on $\mathcal W$ is exactly $(4R_1/3,4R_N/3)$.
If $w,\widetilde w$ are two probability vectors and $w_t=(1-t)w+t\widetilde w$, then the exact concavity identity is
\begin{equation}\label{eq:range-strict-concavity}
 r_{\rm e}(w_t)=(1-t)r_{\rm e}(w)+t r_{\rm e}(\widetilde w)
       +2t(1-t)\int_0^{R_N}(\phi_w-\phi_{\widetilde w})^2\,dr.
\end{equation}
Moving a positive amount of weight from a smaller radius to a larger one, while keeping the other weights fixed, strictly increases both $r_{\rm e}$ and $\beta$.
\end{thm}

\begin{proof}
The function $\phi_w$ is positive for $r>0$, is regular at zero, and equals one for $r\ge R_N$.
Its derivative jump at $R_j$ is $-w_j/R_j=-g_j(w)\phi_w(R_j)$.
Thus it is a zero-energy resonance.
The vector with components $\sqrt{g_j(w)}\phi_w(R_j)$ is a positive $1$-eigenvector of $B(g(w))$, so Perron's theorem gives $\rho(g(w))=1$.
Conversely, the normalized first resonance at a critical $g$ is unique and has the representation in \eqref{eq:inverse-critical-map} with $w_j=g_jR_j\phi(R_j)>0$ and $\sum_jw_j=1$.
This proves bijectivity and smoothness; the same reasoning on the active index set proves the closed-simplex assertion.
All denominators stay positive on the closed simplex, since $\phi_w(R_j)\ge R_1/R_N$.

From $v_j=\sqrt{\beta}\sqrt{g_j}\phi(R_j)$ and $\|v\|=1$,
\[
 \beta^{-1}=\sum_jg_j\phi(R_j)^2
            =\sum_j\frac{w_j\phi(R_j)}{R_j}
            =\sum_{i,j}\frac{w_iw_j}{\max(R_i,R_j)}.
\]
Integration of the piecewise constant derivative gives the same formula for $\int\phi_w'^2$.
The kernel identity, for $x\le y$,
\[
 y+\frac{x^2}{3y}
   =\frac23(x+y)+\frac{(y-x)^2}{3y}
\]
and \eqref{eq:critical-weight-formula} give \eqref{eq:range-variance-weights}.
The integral representation $r_{\rm e}=2\int_0^{R_N}(1-\phi_w^2)\,dr$ gives \eqref{eq:range-strict-concavity}.
Distinct probability vectors produce distinct piecewise linear profiles: their derivative jumps recover each weight.
Hence the last integral is positive whenever $w\ne\widetilde w$, proving strict concavity.
The open simplex is connected, so its continuous image is an interval.
The strict endpoint bounds and the limits at its innermost and outermost vertices show that this interval is exactly the asserted open interval.

Finally, shifting weight from $R_i$ to $R_j>R_i$ decreases $\phi_w$ pointwise, strictly on an interval of positive length.
The integral formula therefore strictly increases $r_{\rm e}$.
For $\beta^{-1}=w^TCw$, $C_{pq}=1/\max(R_p,R_q)$, the derivative under such a shift is $2\sum_qw_q(C_{jq}-C_{iq})\le0$.
It is strictly negative as long as the source weight at $i$ is positive, because the term $q=i$ is then strictly negative.
Integrating over the allowed shift proves the strict increase of $\beta$.
\end{proof}

\begin{cor}[Compatibility of the global coordinates and local curvature]
\label{cor:geometry-compatibility}
Let $w(\varepsilon)$ be a $C^2$ curve in $\mathcal W$ with $w(0)=w_*$, and set $g(\varepsilon)=g(w(\varepsilon))$.
Write $g_*=g(0)$, $a=g'(0)$, and $\phi=\phi_{w_*}$.
If $a\ne0$, then
\begin{equation}\label{eq:global-local-compatibility}
 \sum_{j=1}^N a_j\phi(R_j)^2=0,
 \qquad
 \sum_{j=1}^N g_j''(0)\phi(R_j)^2=-2C_a.
\end{equation}
Consequently, a nonconstant curve of critical couplings whose first derivative is nonzero must have a normal acceleration that compensates the positive tangent curvature.
In particular, its second-order Taylor expansion cannot be a straight tangent line.
\end{cor}

\begin{proof}
The inverse parametrization gives $p(g(\varepsilon))=0$ for all sufficiently small $\varepsilon$.
Differentiate once and twice at zero and use \eqref{eq:tangent-p-derivatives}.
The first derivative is the tangency condition.
The second derivative is
\[
 0=D^2p(g_*)[a,a]+Dp(g_*)[g''(0)]
   =2C_a+\sum_jg_j''(0)\phi(R_j)^2.
\]
The positivity of $C_a$ completes the argument.
\end{proof}

\begin{cor}[Recovery from individual threshold responses]
\label{cor:threshold-recovery}
At a positive first-critical configuration define
\[
 b_j=\lim_{t\downarrow0}\frac{\sqrt{-E_b(g_*+t e_j)}}{t}.
\]
Then $b_j=\phi(R_j)^2$, so $b_N=1$.
These responses uniquely determine all critical strengths.
More explicitly, put $R_0=0$, $u_0=0$, $u_j=\sqrt{b_j}$, and
\[
 s_j=\frac{u_j-u_{j-1}}{R_j-R_{j-1}},\qquad
 s_{N+1}=0.
\]
The reconstruction is
\[
 (g_*)_j=\frac{s_j-s_{j+1}}{u_j},\qquad 1\le j\le N.
\]
Conversely, prescribed positive response data occur for strictly positive first-critical strengths if and only if $u_N=1$ and $s_1>s_2>\cdots>s_N>0$.
\end{cor}

\begin{proof}
Equation~\eqref{eq:critical-gradient} and Theorem~\ref{thm:general-threshold}(iv) give $\kappa(g_*+te_j)=\phi(R_j)^2t+O(t^2)$.
The critical resonance is piecewise linear, and its derivative jump at $R_j$ is $s_{j+1}-s_j=-(g_*)_j u_j$.
This proves the reconstruction.
Conversely, the stated slope conditions define a positive increasing concave piecewise linear profile, constant one after $R_N$.
Its strictly negative derivative jumps give the displayed positive strengths.
The positive profile makes the associated zero-energy Birman-Schwinger eigenvalue the Perron eigenvalue, hence the first-critical condition holds.
\end{proof}

\section{Numerical checks of the new threshold formulas}\label{sec:numerical-main}

The determinant identity and elementary two-shell $s$-wave analysis were included in earlier versions of the present paper.
We therefore restrict the numerical examples to the uniform total spectral-shift estimate and the tangent-direction coefficient of Theorem~\ref{thm:tangent-binding}.
The figures and bounds below are numerical checks, not substitutes for the analytic estimates.
The channel phases use the coupling normalization in \eqref{eq:eigen-phase}; their summation includes angular-momentum multiplicities.

\subsection{A three-shell check of the uniform formula}
\label{sec:three-shell-check}

To test the formula beyond two shells, take $R=(1,2,4)$ and $\widehat g=(0.4,0.25,0.15)$.
We set $g_*=\widehat g/\rho(\widehat g)$, so that $\rho(g_*)=1$.
The numerical values are
\[
 \rho(\widehat g)=1.1482788693,\qquad
 g_*=(0.3483474361,0.2177171475,0.1306302885).
\]
At this critical configuration,
\[
 \beta_*=3.1334767998,\quad h_*=0.2223570205,\quad
 \gamma_*=7.1253802292,\quad \Lambda_*=2.0515963816.
\]
In particular, the critical effective range is $r_{\rm e}=4.1031927632$, inside the bounds $4/3<r_{\rm e}<16/3$.

For $\sigma=-1,0,1$, let $g=(1+\sigma k)g_*$ and use the cubic approximation in \eqref{eq:DN-phase-normal-form}:
\[
 F_{\rm cub}(k,g)=\frac12+\frac1\pi
       \arctan\left(\frac{p(g)}k+\mathcal I(g)k\right)-\frac{4k}\pi.
\]
The zero-energy shell values of $y_0$, the coordinate $p$, and $\mathcal I$ were evaluated by piecewise linear propagation and exact integration on each interval.
At $g_*$, $\mathcal I_*=1.9484036184$, so that $4-\mathcal I_*=\Lambda_*$.
Table~\ref{tab:three-shell-uniform} gives $|\Delta N_{24}(k^2;g)-F_{\rm cub}(k,g)|$.
The partial-wave computation uses double-precision spherical Bessel functions and the coupling-continued eigenvalue arguments in \eqref{eq:eigen-phase}.
The last column checks the predicted cubic order along these three paths.
These are grid-based calculations and do not establish rigorous error bounds.
The uniform estimate is proved in Theorem~\ref{thm:exact-normal-form}.

\begin{table}[!htbp]
\centering
\caption{Errors of the uniform cubic normal form for three shells,
with $g=(1+\sigma k)g_*$ and $\sigma=-1,0,1$.
The last column uses the maximum over the three paths.}
\label{tab:three-shell-uniform}
\begin{tabular}{rrrrr}
\toprule
$k$&$\sigma=-1$&$\sigma=0$&$\sigma=1$&$\max_{\sigma}\mathrm{error}/k^3$\\
\midrule
0.040 & $3.7002\times10^{-4}$ & $3.9311\times10^{-4}$ & $4.0334\times10^{-4}$ & $6.302166$\\
0.020 & $4.7306\times10^{-5}$ & $4.9173\times10^{-5}$ & $4.9385\times10^{-5}$ & $6.173122$\\
0.010 & $5.9785\times10^{-6}$ & $6.1477\times10^{-6}$ & $6.1084\times10^{-6}$ & $6.147651$\\
0.005 & $7.5139\times10^{-7}$ & $7.6849\times10^{-7}$ & $7.5950\times10^{-7}$ & $6.147912$\\
\bottomrule
\end{tabular}
\end{table}

\subsection{The tangent scale}

For $R=(1,3)$, $g_*=(1/2,1/4)$, and $a=(1,-1/4)$, Theorem~\ref{thm:tangent-binding} gives $C_a=1/4$.
Table~\ref{tab:tangent-check} uses $g_\pm=g_*\pm|\varepsilon|a$, $\kappa_\pm=\kappa(g_\pm)$, and $\Delta N_\pm=\Delta N(\varepsilon^4;g_\pm)$.
It checks both signs of the redistribution and the joint scale $k=\varepsilon^2$.
The pole was found using the exact negative-energy matrix, with 70-digit arithmetic to resolve its quadratic smallness.

\begin{table}[!htbp]
\centering
\caption{A straight tangent redistribution for two shells. The pole is
computed with 70-digit arithmetic. The excess uses the high-precision
$s$ wave and regular higher channels through $\ell=12$, at
$k=\varepsilon^2$. The respective limits are $1/4$ and
$1/2+\pi^{-1}\arctan(1/4)=0.5779791304$.}
\label{tab:tangent-check}
\begin{tabular}{rrrrr}
\toprule
$|\varepsilon|$ & $\kappa_-/\varepsilon^2$
& $\kappa_+/\varepsilon^2$
& $\Delta N_-$
& $\Delta N_+$\\
\midrule
0.080 & 0.2310749430 & 0.2711161225 & 0.5685933029 & 0.5808540790\\
0.040 & 0.2402720804 & 0.2602773675 & 0.5741504999 & 0.5801768598\\
0.020 & 0.2450684002 & 0.2550690654 & 0.5762750836 & 0.5792752707\\
0.010 & 0.2475171448 & 0.2525172281 & 0.5771789015 & 0.5786773715\\
\bottomrule
\end{tabular}
\end{table}

\FloatBarrier
\section*{Concluding remarks}

The first-critical strengths of the attractive concentric shell model form a smooth hypersurface with a strictly positive second fundamental form.
A single explicit quadratic form describes its curvature in every nonzero tangent direction, the fourth-order coefficient of the emerging bound-state energy, and the finer spectral-shift crossover scale.
The positivity holds at every positive first-critical configuration and does not require a genericity assumption on the tangent direction.

The global probability-weight parametrization identifies the entire positive critical surface and gives a strictly concave effective-range function with its exact image.
The compatibility identity \eqref{eq:global-local-compatibility} shows how this global description and the local curvature formula fit together.
Individual threshold-pole responses provide a corresponding inverse determination of the strengths.

The exact threshold coordinate yields a cubic-error normal form for the full partial-wave spectral shift and a differentiated estimate uniform in the ratio $p(g)/k$.
It makes the geometric crossover coefficient observable in the normalized spectral excess rather than only in the $s$-wave pole.
The previously derived channel determinant formula and the general threshold mechanism are used as inputs.
The mathematical contribution is the explicit positive finite-shell curvature form, its coefficient-level spectral consequences, and the global inverse description.

The positive weights, the simple first resonance, and fixed shell radii are used in the strict curvature and inverse-parametrization arguments.
The weighted path-graph identity relies on the radial zero-energy Green kernel and need not retain its positive form for sign-changing couplings, excited resonances, or varying support geometry.
The general fourth-order binding mechanism is not specific to shells, while the exact coefficient and the complete inverse classification are model-dependent.
These distinctions delimit the result without claiming a universal curvature law.

\section*{Acknowledgements}
The author thanks the anonymous referees of an earlier submission of this work for their careful reading and constructive comments.
Their suggestions helped clarify the distinction between established results and the contributions of the present paper.

\appendix
\section{Identification of the normalized spectral shift}
\label{app:normalization}

We give the normalization argument omitted from the main text.
Let $h_\ell$ and $h_{0,\ell}$ be the reduced radial operators in $L^2(0,\infty;dr)$, with the regular condition at zero and centrifugal term $\ell(\ell+1)/r^2$.
The free radial Green matrix at the shell radii is $R_iR_j(m_\ell(z))_{ij}$.
By a diagonal similarity, its perturbation determinant is exactly $\det K_\ell(z)$.
The radial resolvent formula and differentiation of the determinant give
\begin{equation}
\label{eq:radial-trace}
\partial_z\log\det K_\ell(z)
 =-\tr\{(h_\ell-z)^{-1}-(h_{0,\ell}-z)^{-1}\}.
\end{equation}
One can also verify \eqref{eq:radial-trace} directly from $\partial_zm_\ell$ and cyclicity of the finite-rank resolvent trace.

Choose $E_*<0$ below the spectra of $H_t:=H_0+t(H-H_0)$ for all $0\le t\le1$.
Such a common lower bound follows from the uniform form bound in \eqref{eq:form}.
Write $F_\ell(z)=\log\det K_\ell(z)$ for the analytic logarithm in the upper energy half-plane whose value at $E_*$ is real.
More explicitly, at $E_*$ it is the real logarithm of the positive determinant continued from $t=0$.
No determinant vanishes for $\Im z>0$ or for $z=E_*$ along this coupling path, so its positive-energy boundary logarithm is $\mathcal L_\ell(k)$ from \eqref{eq:coupling-log}.

Consider the subtracted series
\begin{equation}
\label{eq:subtracted-log}
F(z)=\sum_{\ell=0}^\infty(2\ell+1)
           \{F_\ell(z)-F_\ell(E_*)\}.
\end{equation}
The subtraction is important: the unsubtracted real parts generally have a divergent angular-momentum sum.
We record the required large-angular-momentum estimates more explicitly.
Fix a compact complex energy set $K$ in the upper energy half-plane, together with its boundary values on a compact positive-energy interval.
For $\ell\ge2$, the spherical Bessel expansions give, uniformly for $z\in K$,
\begin{equation}\label{eq:angular-diag-expansion}
 (m_\ell(z))_{ii}
  =\frac{1}{(2\ell+1)R_i}
   +\frac{2zR_i}{(2\ell-1)(2\ell+1)(2\ell+3)}
   +O_K((1+\ell)^{-5}).
\end{equation}
The derivative of the remainder is also $O_K((1+\ell)^{-5})$ on a slightly smaller compact set.
Indeed, the product expansions of the spherical Bessel functions give
\begin{equation}\label{eq:angular-product-expansion}
 -t j_\ell(t)y_\ell(t)=\frac{1}{2\ell+1}
 \left(1+\frac{2t^2}{(2\ell-1)(2\ell+3)}
 +O_K((1+\ell)^{-4})\right).
\end{equation}
The coefficient of $t^4$ inside the parentheses is
$6/[(2\ell-3)(2\ell-1)(2\ell+3)(2\ell+5)]$.
The higher terms and their first derivatives in $t^2$ satisfy the stated remainder bound uniformly for bounded $t$.
The term $t j_\ell(t)^2$ from the outgoing Hankel function is superexponentially small in $\ell$ on such sets.
Substituting $t=\sqrt z R_i$ into \eqref{eq:angular-product-expansion} proves \eqref{eq:angular-diag-expansion}, including the bound for its energy derivative.
For $i\ne j$ put $q_{ij}=R_</R_><1$.
The same expansions at the two distinct radii give
\begin{equation}\label{eq:angular-offdiag-bound}
 \left|(m_\ell(z))_{ij}-(m_\ell(E_*))_{ij}\right|
  +\left|\partial_z(m_\ell(z))_{ij}\right|
 \le C_K\frac{q_{ij}^{\ell}}{(1+\ell)^2}.
\end{equation}
By \eqref{eq:angular-diag-expansion} and \eqref{eq:angular-offdiag-bound},
\begin{equation}\label{eq:angular-matrix-bound}
 \sup_{z\in K}\left(
 \|m_\ell(z)-m_\ell(E_*)\|
 +\|\partial_zm_\ell(z)\|\right)
 \le C_K(1+\ell)^{-3}.
\end{equation}
Moreover, $\|m_\ell(z)\Theta\|\le C_K/(1+\ell)$.
For sufficiently large $\ell$, the matrices $I+m_\ell(z)\Theta$ and $I+m_\ell(E_*)\Theta$ are uniformly invertible.
The elementary identity for the difference of two matrix logarithms, or integration along the segment joining these matrices, gives
\begin{equation}\label{eq:angular-log-bound}
 \sup_{z\in K}\left(
 |F_\ell(z)-F_\ell(E_*)|+|F_\ell'(z)|\right)
 \le C_K(1+\ell)^{-3}.
\end{equation}
The right-hand side remains summable after multiplication by the degeneracy $2\ell+1$.
Thus the subtracted series \eqref{eq:subtracted-log} and its energy derivative converge locally uniformly in the upper half-plane.
The same bounds hold for boundary values on compact positive-energy intervals, because the Bessel products are analytic in the wave number near each positive real $k$.
This justifies taking the boundary imaginary part and the derivative term by term.
For the finitely many small angular momenta, invertibility at positive energy follows from Proposition~\ref{thm:Sell}.
The arguments rely only on the fixed finite set of distinct radii and the finite strengths, not on a cutoff in angular momentum.

The direct-sum decomposition of the trace-class resolvent difference, together with \eqref{eq:radial-trace}, now gives
\begin{equation}
\label{eq:F-derivative}
F'(z)=-\tr\{(H_\alpha-z)^{-1}-(H_0-z)^{-1}\}.
\end{equation}
The canonical resolvent trace formula is $$ \tr\{(H_\alpha-z)^{-1}-(H_0-z)^{-1}\} =-\int_\R\frac{\xi(E)}{(E-z)^2}\,dE.
$$ The normalization is $F(E_*)=0$.
Integrating the derivative identity from $E_*$ therefore fixes the integration constant.
It follows that
\begin{equation}
\label{eq:F-integral}
F(z)=\int_\R\xi(E)
       \left\{\frac1{E-z}-\frac1{E-E_*}\right\}\,dE.
\end{equation}
The difference in braces is integrable against the weighted spectral shift; the second denominator has no singularity on its support.
At a Lebesgue point of $\xi$, $\Im F(E+i0)=\pi\xi(E)$.
Since every $F_\ell(E_*)$ is real, the boundary imaginary part of \eqref{eq:subtracted-log} equals $\sum_\ell(2\ell+1)\Im\mathcal L_\ell(\sqrt E)$.
This proves \eqref{eq:SSF-real} with the canonical normalization and no remaining integer.
Lemma~\ref{lem:tail} then supplies its continuously differentiable positive-energy representative.

\end{document}